\documentclass[12pt]{article}
\usepackage{fullpage}
\usepackage[T1]{fontenc}

\usepackage{authblk}
\usepackage{amsthm}

\newcommand{\eps}{\varepsilon}

\newcommand{\MST}{\emph{MST}}

\newcommand{\light}{\emph{lightness}}

\newcommand{\diam}{\emph{diam}}

\newcommand{\Oish}{\widetilde{O}}

\usepackage{amsfonts}
\usepackage{amsmath}
\usepackage{todonotes}
\usepackage{hyperref}
\usepackage{cleveref}
\usepackage{color}

\usepackage{etoolbox}

\usepackage{breqn}

\newtheorem{theorem}{Theorem}[]
\newtheorem{lemma}{Lemma}[]

\newtheorem{definition}[lemma]{Definition}

\title{Lightweight Near-Additive Spanners}
\author[1]{Yuval Gitlitz}
\author[1]{Ofer Neiman} 
\author[1]{Richard Spence}

\affil[1]{Ben-Gurion University of the Negev. Emails: \texttt{gitlitz@post.bgu.ac.il}, \texttt{neimano@cs.bgu.ac.il}, \texttt{rcspence1234@gmail.com}}

\date{\today}
\begin{document}

\maketitle

\begin{abstract}

An $(\alpha,\beta)$-spanner of a weighted graph $G=(V,E)$, is a subgraph $H$ such that for every $u,v\in V$, $d_G(u,v) \le d_H(u,v)\le\alpha\cdot d_G(u,v)+\beta$. The main parameters of interest for spanners are their size (number of edges) and their lightness (the ratio between the total weight of $H$ to the weight of a minimum spanning tree).

In this paper we focus on near-additive spanners, where $\alpha=1+\eps$ for arbitrarily small $\eps>0$. 
We show the first construction of {\em light} spanners in this setting. Specifically, for any integer parameter $k\ge 1$, we obtain an $(1+\eps,O(k/\eps)^k\cdot W(\cdot,\cdot))$-spanner with lightness $\Oish(n^{1/k})$ (where $W(\cdot,\cdot)$ indicates for every pair $u, v \in V$ the heaviest edge in some shortest path between $u,v$). In addition, we can also bound the number of edges in our spanner by $O(kn^{1+3/k})$. 


\end{abstract}

\section{Introduction}
Given an undirected, edge-weighted graph $G=(V,E,w)$ with $|V|=n$, and parameters $\alpha \ge 1$ and $\beta \ge 0$, a subgraph $H=(V,E',w)$ is an \emph{$(\alpha,\beta)$-spanner} if
\begin{equation} \label{eq:spanner-inequality}
    d_G(x,y) \le d_H(x,y) \le \alpha \cdot d_G(x,y) + \beta
\end{equation}
for all vertex pairs $x,y \in V$, where $d_G(x,y)$ denotes the distance between $x$ and $y$ in $G$. Spanners and related graph constructs such as emulators (which satisfy~\eqref{eq:spanner-inequality} but need not be subgraphs of the given graph) have gained attention for their role in distributed computing, approximation of shortest paths and graph compression; see the recent survey~\cite{ahmed2020survey}. One of the primary goals is to determine the tradeoffs between the multiplicative stretch $\alpha$, the additive error $\beta$, and the size of the spanner $H$ as measured by $|E'|$. An additional (sometimes even more desirable) goal is to construct \emph{lightweight} spanners whose total weight is bounded. Of course, one must pay at least the weight of a minimum spanning tree (MST) to simply ensure $H$ is connected; thus, the lightness of a spanner $H$ is measured as the ratio of the edge weight of $H$ to the edge weight of the MST of $G$. For algorithmic applications, it is also desired to construct these spanners efficiently.

Spanners were first studied in the 1980s~\cite{peleg1989graph} with purely multiplicative stretch, i.e., $\beta = 0$. In \cite{althofer1993sparse} it was shown that for all integers $k \ge 1$, all weighted graphs have a $(2k-1,0)$-spanner with size $O(n^{1+1/k})$ and lightness $O(\frac{n}{k})$ via a simple greedy algorithm. The size bound is optimal assuming Erd\H{o}s' girth conjecture~\cite{erdos1963extremal}.
Soon thereafter, \cite{chandra1992sparse} showed for any $0<\eps<1$, the same greedy algorithm yields a $((2k-1)(1+\eps),0)$-spanner with lightness $O_\eps(kn^{1/k})$ and size $O_\eps(n^{1+1/k})$. \footnote{The notation $O_{\eps}(\cdot)$ hides $(\frac{1}{\eps})^{O(1)}$ factors.} Researchers also sought to understand the tradeoffs if one desires a tiny amount $1+\eps$ of multiplicative error, also known as {\em near-additive spanners}. For unweighted graphs, in a sequence of works, \cite{elkin2004spanner,P09,EN19,bodwin2015sparse,abboud2017hierarchy}, $(1+\eps, \beta)$-spanners of size $\approx O(n^{1+1/k})$ where $\beta = O\left(\left(\frac{\log k}{\eps}\right)^{\log k}\right)$ were devised. In \cite{abboud2017hierarchy} it was shown that the additive stretch of near-additive spanners of such size must be $\Omega_k\left(\frac{1}{\eps}\right)^{\log k}$\footnote{The notation $\Omega_{k}(\cdot)$ hides $k^{O(1)}$ factors.}. 
Near-additive spanners were also studied for weighted graphs. Here, the additive error must depend linearly on the maximum edge weight $W_{max}$, or the ``local'' maximum edge weight $W(x,y)$ along a shortest $x$-$y$ path $P_{x,y}$. The latter is typically more preferred as it is possible for $W_{max} \gg W(x,y)$; in this case, we use $W(\cdot,\cdot)$ to indicate that the additive error for each vertex pair $x,y$ is in terms of $W(x,y)$. In both cases, the best one can hope for is to match the corresponding bounds for unweighted graphs. Indeed, \cite{elkin2019almost} showed that all weighted graphs admit $(1+\eps, \left(\frac{\log k}{\eps}\right)^{O(\log k)}W(\cdot,\cdot))$ spanners of size $O(kn + n^{1 + 1/k})$.

In the past decade, there has been an surge of interest in light spanners, and algorithms for their construction were devised in various settings. For general graphs, light spanners were studied in several works \cite{elkin2014light,chechik2016light,LS23,Bodwin23}, culminating in lightness $O(n^{1/k}/\eps)$ (with stretch $(2k-1)(1+\eps)$ and size $O(n^{1+1/k}/\eps)$).
For Euclidean pointsets of dimension $q$, \cite{S91,V91} showed a $(1+\eps)$-spanner of size $n\cdot\eps^{-O(q)}$ and lightness $\eps^{-O(q)}$. For high dimensional pointsets, \cite{FN22} devised $O(k)$-spanners with lightness $\tilde{O}(n^{1/k^2})$.\footnote{The notation $\Oish(f(n))$ hides polylogarithmic factors in $n$.} The low dimensional results were extended for metric spaces of doubling dimension $ddim$: \cite{Gottlieb15,BLW19,FS20} constructed a $(1+\eps)$-spanner with $n\cdot\eps^{-O(ddim)}$ edges and lightness $\eps^{-O(ddim)}$. 
For any graph excluding a fixed minor, \cite{BLW17} devised $(1+\eps)$-spanners with lightness $\tilde{O}(\eps^{-3})$. Light prioritized spanners were shown by \cite{BFN19}, which were used to get an almost MST with constant average distortion.\footnote{A prioritized spanner receives an arbitrary ranking of the points, and should obtain better stretch for high ranking points.} Recently, \cite{FGN24} devised light reliable spanners.\footnote{A reliable spanner is susceptible to massive vertex failures $B\subseteq V$, and still provide meaningful guarantees for all vertex pairs in $V\setminus B^+$, where $B^+$ is only slightly larger than $B$.}
However, despite all these efforts, no previous work gave a meaningful bound on the lightness of near-additive spanners.


\subsection{Our Results}

In this paper we give the first construction of {\em light} near-additive spanners for general graphs. The specific bounds are described in the following theorem.

\begin{theorem} \label{thm:spanner-near-additive}
Let $G$ be a weighted graph, let $0 < \eps < 1$, and let $k$ be a positive integer. Then $G$ has a $\left(1+\eps,O(\frac{k}{\eps})^{k} \cdot W(\cdot,\cdot)\right)$-spanner of size $O(kn^{1 + 3/k})$ and lightness $\Oish(\frac{n^{1/k}}{\eps})$.
\end{theorem}

We remark that the dependence of $\beta$, the additive stretch, on $k$ is somewhat inferior compared to the result of \cite{elkin2019almost}. However, their construction has unbounded lightness. There is an efficient randomized algorithm to construct the spanners of Theorem~\ref{thm:spanner-near-additive}. Note that the lower bound of \cite{abboud2017hierarchy} also holds for near-additive spanners of bounded lightness.


Additionally, if $W_{max} \ge \sqrt{w(\MST(G))}$, we can construct a $(1+\eps, 2(1+\eps)W_{max})$-spanner of size $O(n^{3/2})$ and lightness $O(\frac{n^{1/2}}{\eps})$.


\subsection{Technical Overview}

One of the most popular techniques to design sparse spanners is the method devised by Thorup-Zwick \cite{TZ05}. The basic idea is to sample a sequence of sets $V=A_0\supseteq A_1\supseteq\dots\supseteq A_k$, and define the bunch of $v\in A_i\setminus A_{i+1}$ as all the vertices closer to $v$ than the nearest vertex in $A_{i+1}$ (which is called the pivot). The spanner is constructed by taking all the shortest paths from $v$ to its bunch, and the path to its pivot. There are certain variations of this basic construction that yield the best known $(\alpha,\beta)$-spanners \cite{NS22}, in essentially all the regimes of multiplicative stretch $\alpha$.

Our main technical innovation is an adaptation of this framework that ensures the lightness of the resulting spanner. Specifically, we construct a hierarchy of {\em nets} in all distance scales of the form $2^i$, and define for each vertex a net-point {\em representative} in each such scale. Now, we replace the connections from each $v$ to each vertex $u$ in its bunch, by connecting $v$ to $u$'s representative in the net of scale $\approx \eps d(u,v)$. In fact, we shall do this only for vertices $u$ which are closer than $\frac{1-\eps}{2}$ times the distance of $v$ to its pivot. Using this variant, we are able to show that any representative participates in $\approx n^{1/k}$ such paths (the basic idea is that the set of all vertices $v$ who connect to a certain representative $x$, must all lie in a single bunch, whose size is bounded). Then we apply a known relation between nets and the weight of the MST (see Lemma~\ref{lem:MSTnet}) to prove the lightness of these paths.

In addition, to complete the $v-u$ spanner path, we need to guarantee that $u$ is connected to each of its representatives. Directly adding a path from $u$ to the nearest net-point in all scales turns out to be too costly in terms of lightness. Instead, we connect $u$ only to $\log\frac{1}{\eps}$ representatives, and show that these connections are light enough, yet they imply there exist ``sufficiently good'' representatives in all scales. Finally, 
we replace the connections to the pivots by taking Shallow Light Trees (SLT) rooted at the pivots. 
The (quite technical) analysis of the stretch has to take into consideration all these various changes, and we indeed show how they only incur a small factor in the additive stretch.

The selection of the sets $A_i$ in previous works, and also here, is done randomly. Each vertex in $A_i$ advances to $A_{i+1}$ either with uniform $n^{-1/k}$ probability, or with exponentially decreasing $\approx n^{-2^i/k}$ probability (in the latter case there are only $\log k$ sets). The improved dependence of the additive stretch $\beta\approx\left(\frac{\log k}{\eps}\right)^{\log k}W(\cdot,\cdot)$ in \cite{elkin2019almost}, is achieved using the latter probability choice. However, a key ingredient in our lightness analysis, that each representative has roughly $n^{1/k}$ paths added to it, crucially relies on the uniform probability choice. For this reason we can only achieve $\beta\approx \left(\frac{k}{\eps}\right)^{k}W(\cdot,\cdot)$, and we leave open the question of obtaining an improved additive stretch.

\subsection{Related Work}

Pure additive spanners (with $\alpha = 1$) have been widely studied in recent years, mostly in unweighted graphs. It is known that all unweighted graphs have $(1,2)$-spanners of size $O(n^{3/2})$~\cite{aingworth99fast}, $(1,4)$-spanners of size $\Oish(n^{7/5})$~\cite{chechik2013new}, and $(1,6)$-spanners of size $O(n^{4/3})$~\cite{baswana2010additive,woodruff2010additive}. Moreover, the upper bound of $O(n^{4/3})$ cannot be improved for pure additive spanners, even if one allows $n^{o(1)}$ additive error~\cite{abboud2016spanner}. In ~\cite{DBLP:conf/stoc/TanZ23}, linear-size spanners  were shown with polynomial additive stretch $O(n^{0.403})$.

For weighted graphs~\cite{ahmed2021weighted,elkin2021improved} show constructions of pure additive $(1, O(1) \cdot W(\cdot,\cdot))$-spanners with size guarantees that nearly match their unweighted counterparts. The only previous construction of light spanner with additive stretch was given by~\cite{ahmed2021weighted}, who devised a  $(1,(4+\eps)W(\cdot,\cdot))$-spanner with lightness $O_{\eps}(n^{2/3})$. This result has pure additive stretch, and thus is incomparable to our results.

\section{Preliminaries}
For a graph $G=(V,E)$ and subset $S \subseteq V$, we ordinarily denote $d_G(v,S) = \min\limits_{x \in S} d_G(v,x)$.  We use $d(u,v)$ and $d(u,S)$ as shorthand for $d_G(u,v)$ and $d_G(u,S)$ respectively, where $G$ is the input graph. For all constructions, we assume, w.l.o.g., $w(\MST(G)) = n$; in other words, the average weight of an MST edge is $\approx 1$. This assumption does not lose generality as one can simply scale the edge weights, as well as $W_{max}$ and $W(\cdot,\cdot)$. Also, let $[k] = \{1,2,3,\ldots,k\}$.
\begin{definition}[Shallow-light tree (SLT)] \label{def:slt}
Given a weighted graph $G=(V,E)$, parameters $\alpha, \gamma \ge 1$, and a root vertex $s \in V$, a spanning tree $T=(V,E')$ is an $(\alpha,\gamma)$-\emph{shallow light tree} (or $(\alpha,\gamma)$-SLT) rooted at $s$ if the following are true:
\begin{itemize}
\item $d_T(s,x) \le \alpha\cdot d(s,x)$ for all $x \in V$, and
\item $w(T) \le \gamma \cdot w(\MST(G))$.
\end{itemize}
\end{definition}
We use the following result by Khuller et al.~\cite{khuller1995balancing}:
\begin{theorem}[\cite{khuller1995balancing}]
Given a weighted graph $G$, a root vertex $s \in V$, and $\eps > 0$, there exists a polynomial-time constructible $(1+\eps, O(\frac{1}{\eps}))$-SLT rooted at $s$.
\end{theorem}

\begin{definition}[$\Delta$-net]\label{def:delta-net}
Given a weighted graph $G=(V,E)$ and a real number $\Delta > 0$, a subset of vertices $N \subseteq V$ is a \emph{$\Delta$-net} if the following hold:
\begin{itemize}
\item $d(x,N) \le \Delta$ for all $x \in V$, and
\item $d(u,v) > \Delta$ for all distinct $u,v \in N$.
\end{itemize}
\end{definition}
A $\Delta$-net $N$ can be computed via the following simple greedy algorithm: while there exists a vertex whose distance to $N$ is greater than $\Delta$, choose such a vertex and add it to $N$. The following lemma relates nets to the weight of the MST.
\begin{lemma}\label{lem:MSTnet}
For any $\Delta$-net $N$ of a graph $G$ with $|N| \ge 2$ vertices, we necessarily have $|N|\Delta\le 2w(\MST(G))$. 
\end{lemma}
\begin{proof}
Let $T$ be the MST of $G$.
Note that balls of radius $\frac{\Delta}{2}$ centered at each vertex $v \in N$ are pairwise disjoint, as every pair of vertices in $N$ is of distance greater than $\Delta$. For each $v\in N$, consider the weight of edges in $T$ that lie in $B=B(v,\frac{\Delta}{2})$ (if an edge $\{x,y\}\in T$ leaves the ball, suppose $x\in B$ and $y\notin B$, we consider the relative weight inside, that is $\frac{\Delta}{2}-d(v,x)$). The observation is that every such ball contains at least $\frac{\Delta}{2}$ of the weight of $T$, which is unique to it. Therefore, $w(T)\ge |N|\frac{\Delta}{2}$.
\end{proof}



\section{A lightweight \texorpdfstring{$(1+\eps, O(\frac{k}{\eps})^k)$}~~near-additive spanner}\label{section:lightweight}



In this section we prove our main result, Theorem~\ref{thm:spanner-near-additive}. The construction of the spanner has three phases; the first phase constructs a hierarchical sequence of $O(\log n)$ $\Delta$-nets (Def.~\ref{def:delta-net}), and adds a small set of lightweight paths connecting vertices to their nearby net points, which will be used to define a set of \emph{representative} vertices. The second phase stratify the vertices into $k$ levels, then add paths from each vertex $u$ to the representatives of its bunch. The third and final phase uses SLTs for the (approximate) pivot connections.

\subsubsection*{First phase: $\Delta$-nets}
In the first phase, we will construct a nested sequence of $\Delta$-nets (Def.~\ref{def:delta-net}), then add shortest paths from each vertex to the closest vertices in some of the $\Delta$-nets. Specifically, let $V=N_{-1}\supseteq N_0 \supseteq \ldots \supseteq N_{ \log_2 n}$, where for each $i\ge 0$, $N_i$ is a $2^i$-net. In particular, $N_{\log_2 n}$ is an $n$-net so $|N_{\log_2 n}| = 1$ (as we assumed  $w(\MST(G)) = n$, so $\diam (G) \le n$), and $N_0$ is a 1-net. The nets can be constructed using the greedy algorithm following Def.~\ref{def:delta-net}: let $N_{\log n}$ consist of a single vertex, and for $i = \log_2 n-1, \ldots, 0$, use the greedy algorithm to add vertices to $N_{i+1}$ until a $2^i$-net $N_i$ is constructed. 

\paragraph{Constructing $H_0$.} For each $i\ge -1$, each $v \in N_i \setminus N_{i+1}$, and for all $j \in \{i+1, \ldots, i + \lceil\log_2 \frac{1}{\eps}\rceil\}$, add a shortest path to $H_0$ from $v$ to the closest vertex in $N_j$ (if $j > \log_2 n$, simply do nothing). Lastly add $H_0$ to $H$.

\begin{lemma}\label{lemma:representative}
For all $v \in V$ and for all $i \in \{0,\ldots,\log_2 n\}$, there exists $x \in N_i$ such that $d_{H_0}(v,x) \le (1+2\eps)2^i$.
\end{lemma}
\begin{proof}

The idea is to follow a sequence of paths in $H_0$ starting from $v$, where each path leads to a vertex $t:=\lceil \log_2 \frac{1}{\eps}\rceil$ levels higher, until we reach a vertex at level $i$.
Write $i=a+bt$ for some integers $0\le a<t$ and $b\ge 0$. First we take the path in $H_0$ from $v$ to the closest vertex $v_0\in N_a$. Then, we take additional $b$ paths, the $j$th one, for $1\le j\le b$, will be from $v_{j-1}\in N_{a+(j-1)t}$ to the closest $v_j\in N_{a+jt}$. All these paths exist in $H_0$, since every vertex in a net is connected to a vertex in all the $t$ nets above it.

Since $N_{a+jt}$ is a $2^{a+jt}$-net, the total length of the concatenation of these paths, that lead from $v$ to a vertex $v_b\in N_i$, is at most
\[
\sum_{j=0}^b 2^{a+jt}=2^i\cdot\sum_{j=0}^b2^{-jt}\le 2^i\cdot\sum_{j=0}^b\eps^j \le 2^i\cdot\left(1+2\eps\right)~,
\]
where we use that $2^{-t}\le \eps$, and the last inequality uses that $\eps<1/10$.

 \end{proof} 
Following Lemma~\ref{lemma:representative}, we define for each $v \in V$ and $i \in \{0,\ldots, \log_2 n\}$ the \emph{level-$i$ representative} $r_i(v)$ to be a vertex $x \in N_i$ such that $d_{H_0}(v,x) \le (1+2\eps)2^i$. These representatives will be used in the second phase of the construction. Note that if $v \in N_i$, then $r_i(v) = v$. We also define for $i<0$, $r_i(v)=v$. 

\subsubsection*{Second phase: Bunches and Representatives}
We describe the second phase of the construction. This phase is similar to those given in~\cite{thorup2006spanner,elkin2019almost}, in which the main idea is to construct a hierarchy of vertex sets $V = A_0 \supseteq A_1 \supseteq A_2 \supseteq \ldots \supseteq A_k \supseteq A_{k+1} = \emptyset$, where each vertex in $A_{i-1}$ is included in $A_i$ with a certain probability. The sets $A_i$ are unrelated to the vertices $N_i$ in the $2^i$-net. The main difference is that instead of connecting each vertex $u$ to its bunch, we connect it to a carefully chosen set of representatives of a bunch whose radius is slightly less than half the distance to $u$'s pivot (unlike~\cite{elkin2019almost}, which used \emph{half-bunches}). 

Let $V = A_0 \supseteq A_1 \supseteq A_2 \supseteq \ldots \supseteq A_k \supseteq A_{k+1} = \emptyset$ be sets of vertices. For $i \in [k]$, $A_i$ is obtained by sampling each vertex in $A_{i-1}$ independently with probability $n^{-1/k}$. 
Given $u \in V$ and $i \in \{0,\ldots,k\}$, the level-$i$ \emph{pivot} of $u$, denoted $p_i(u)$, is a vertex $x \in A_i$ such that $d(u,x) = d(u,A_i)$; note that if $u \in A_i$, then $p_i(u) = u$. Given $u \in A_i \setminus A_{i+1}$, and a real number $\delta > 0$, let the $\delta$-\emph{bunch} centered at $u$ be defined as follows:
\[B_{r}(\delta) = \begin{cases}
\{v \in A_i : d(u,v) < \delta \cdot d(u, p_{i+1}(u))\} & 0 \le i \le k-1 \\
A_{k} & i = k
\end{cases}\]
\paragraph{Adding paths to representatives.} We now add paths to the spanner $H$ as follows: for all vertices $u \in V$, and for all $v \in B_{\frac{1-\eps}{2}}(u)$, let $j \le\log n$ be the unique integer such that $\frac{\eps}{8} \cdot d(u,v) \le 2^j < \frac{\eps}{4} \cdot d(u,v)$. We connect $u$ to the level-$j$ representative $r_j(v)$ of $v$ via a shortest path.
Note that if $j<0$, we add a direct $u-v$ shortest path to $H$, since we defined $r_j(v)=v$ for such $j$. On the other hand, we cannot have $j>\log n$, because by our assumption on the MST weight, $d(u,v)\le \diam(G)\le n$. Denote by $r(u,v)$ the level-$j$ representative of $v$ that $u$ connected to.

\subsubsection*{Third phase: SLTs}

Lastly, for each $i \in [k]$, consider a ``virtual'' root vertex $s_i$ which is connected to the vertices in $A_i$ with (virtual) zero-weight edges, and add the edges of a $\left(1+\eps, O\left(\frac{1}{\eps}\right)\right)$-SLT rooted at $s_i$ to the spanner (excluding the zero-weight edges). Let $H$ be the resulting subgraph. This final step effectively creates $k$ forests of SLTs rooted at $A_1, \dots, A_{k}$, in which for every vertex $u$ and $i \in \{0,\ldots,k\}$, we have  $d_H(u,A_i) \le (1+\eps)d(u,A_i) = (1+\eps)d(u,p_i(u))$. Note that, unlike the constructions in~\cite{thorup2006spanner,elkin2019almost}, our construction does not necessarily connect centers $u$ to their pivots $p_i(u)$ directly, but instead connects to ``approximate'' pivots in $A_i$ which are $1+\eps$ times the distance $d(u,A_i)$ to $u$'s actual pivot. We denote this approximate pivot by $p_{i}'(u)$, such that $p_i'(u) \in A_i$, and $d_H(u,p_{i}'(u)) \le (1+\eps)d(u,p_i(u))$ (see Figure~\ref{fig:connections}).

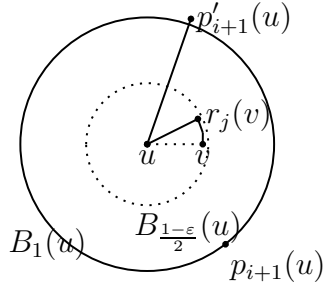
\begin{figure}
\centering

\tikzset{every picture/.style={line width=0.75pt}} 

\begin{tikzpicture}[x=0.75pt,y=0.75pt,yscale=-0.5,xscale=0.5]

\draw  [dash pattern={on 0.84pt off 2.51pt}] (271.5,143.1) .. controls (271.5,109.08) and (299.08,81.5) .. (333.1,81.5) .. controls (367.12,81.5) and (394.7,109.08) .. (394.7,143.1) .. controls (394.7,177.12) and (367.12,204.7) .. (333.1,204.7) .. controls (299.08,204.7) and (271.5,177.12) .. (271.5,143.1) -- cycle ;
\draw    (333.1,143.1) -- (384.2,117.6) ;
\draw [shift={(384.2,117.6)}, rotate = 333.48] [color={rgb, 255:red, 0; green, 0; blue, 0 }  ][fill={rgb, 255:red, 0; green, 0; blue, 0 }  ][line width=0.75]      (0, 0) circle [x radius= 2.34, y radius= 2.34]   ;
\draw [shift={(333.1,143.1)}, rotate = 333.48] [color={rgb, 255:red, 0; green, 0; blue, 0 }  ][fill={rgb, 255:red, 0; green, 0; blue, 0 }  ][line width=0.75]      (0, 0) circle [x radius= 2.34, y radius= 2.34]   ;
\draw  [dash pattern={on 0.84pt off 2.51pt}]  (333.1,143.1) -- (389.1,143.1) ;
\draw [shift={(389.1,143.1)}, rotate = 0] [color={rgb, 255:red, 0; green, 0; blue, 0 }  ][fill={rgb, 255:red, 0; green, 0; blue, 0 }  ][line width=0.75]      (0, 0) circle [x radius= 2.34, y radius= 2.34]   ;
\draw   (205.13,143.1) .. controls (205.13,72.42) and (262.42,15.13) .. (333.1,15.13) .. controls (403.78,15.13) and (461.08,72.42) .. (461.08,143.1) .. controls (461.08,213.78) and (403.78,271.08) .. (333.1,271.08) .. controls (262.42,271.08) and (205.13,213.78) .. (205.13,143.1) -- cycle ;
\draw    (389.1,143.1) -- (389.2,131.6) ;
\draw    (389.2,131.6) -- (384.2,117.6) ;
\draw    (412.1,244.1) ;
\draw [shift={(412.1,244.1)}, rotate = 0] [color={rgb, 255:red, 0; green, 0; blue, 0 }  ][fill={rgb, 255:red, 0; green, 0; blue, 0 }  ][line width=0.75]      (0, 0) circle [x radius= 2.34, y radius= 2.34]   ;
\draw    (333.1,143.1) -- (377.2,16.6) ;
\draw [shift={(377.2,16.6)}, rotate = 289.22] [color={rgb, 255:red, 0; green, 0; blue, 0 }  ][fill={rgb, 255:red, 0; green, 0; blue, 0 }  ][line width=0.75]      (0, 0) circle [x radius= 2.34, y radius= 2.34]   ;
\draw [shift={(333.1,143.1)}, rotate = 289.22] [color={rgb, 255:red, 0; green, 0; blue, 0 }  ][fill={rgb, 255:red, 0; green, 0; blue, 0 }  ][line width=0.75]      (0, 0) circle [x radius= 2.34, y radius= 2.34]   ;

\draw (333.1,146.5) node [anchor=north] [inner sep=0.75pt]    {$u$};
\draw (389.2,114.6) node [anchor=west] [inner sep=0.75pt]    {$r_{j}(v)$};
\draw (389.1,146.5) node [anchor=north] [inner sep=0.75pt]    {$v$};
\draw (190,225.4) node [anchor=north west][inner sep=0.75pt]    {$B_{1}(u)$};
\draw (317,207.4) node [anchor=north west][inner sep=0.75pt]    {$B_{\frac{1-\eps}{2}}(u)$};
\draw (414.1,247.5) node [anchor=north west][inner sep=0.75pt]    {$p_{i+1}(u)$};
\draw (380.2,16.6) node [anchor=west] [inner sep=0.75pt]    {$p'_{i+1}(u)$};

\end{tikzpicture}
\caption{Illustration of the paths added during the three phases. The dotted and solid circles represent the $\frac{1-\eps}{2}$-bunch, and the 1-bunch of $u \in A_i \setminus A_{i+1}$, respectively. The first phase connects $v$ to its representative $r_j(v)$, the second phase connects $u$ to $r_j(v)$, which is one of the representatives of its $\frac{1-\eps}{2}$-bunch, and the third phase connects $u$ to $p'_{i+1}(u)$, the approximate pivot, via the SLT rooted at $s_{i+1}$.}
\label{fig:connections}
\end{figure}

\subsection{Stretch of the spanner}\label{subsection:stretch}
In this section, we prove that $H$ is an $(\alpha,\beta)$-spanner with $\alpha$ and $\beta$ as given in Theorem~\ref{thm:spanner-near-additive}. The proofs in this section follow along the lines of the proofs in~\cite{elkin2019almost} regarding sparse, near-additive emulators and spanners; the main difference is that by using paths from vertices to ``approximate'' pivots using SLTs, and by using paths from vertices to the representatives of their $\frac{1-\eps}{2}$-bunches (instead of direct paths to their $\frac{1}{2}$-bunch), we incur slightly weaker upper bounds on the additive error. We first show the following lemma:
\begin{lemma} \label{lemma:distance-in-bunch}
Let $u \in V$ and $v \in B_{\frac{1-\eps}{2}}(u)$, then $d_H(u,v) \le (1+\eps)d(u,v)$.
\end{lemma}
\begin{proof}
Let $r=r(u,v)$ (recall that $r(u,v)$ is the level-$j$ representative of $v$ that $u$ connected to).
If $r=v$ then we added a shortest $u-v$ path to $H$, therefore the statement holds. Otherwise, $r$ is chosen as $r_j(v)$ for $0\le j\le \log n$ satisfying $\frac{\eps}{8} d(u,v) \le 2^j < \frac{\eps}{4} d(u,v)$. Note that in phase 2 we added a shortest $u-r$ path to $H$. By Lemma~\ref{lemma:representative} we have that $d_H(v,r)\le(1+2\eps)2^j<\frac{\eps}{2} d(u,v)$, thus

\begin{align}
d_H(u,v) &\le d_H(u,r)+d_H(v,r) \tag*{}\\
&= d(u,r) +d_H(v,r)\le d(u,v) + d(v,r)+d_H(v,r) \tag*{}\\
&\le d(u,v) + 2d_H(v,r) < \left(1+\eps\right)d(u,v). \label{eq:dist-to-rep}
\end{align}
 \end{proof}


The following lemma will be useful for pairs that are sufficiently far apart.

\begin{lemma}\label{lemma:distance-long}
Fix $0 < \eps < \dfrac{1}{10}$, let $\Delta > 7$, let $0 \le i \le k$ and $x,y \in V$ such that $d(x,y) \ge (3\Delta)^i W(x,y)$. Let $\eta_i = 1+\eps+\frac{14}{\Delta-7}i$ and $\zeta = \frac{3\Delta}{\Delta-7}$. Then at least one of the following statements is true:
\begin{itemize}
\item $d_H(x,y) \le \eta_i d(x,y)$
\item $d_H(x,A_{i+1}) \le \zeta d(x,y)$.
\end{itemize}
\end{lemma}

The proof of \Cref{lemma:distance-long} is in
\cref{sec:lemma:distance-long:proof}.

The next lemma is useful for vertex pairs which are relatively nearby.
\begin{lemma}\label{lemma:distance-short}
Let $0 < \eps < \dfrac{1}{10}$, let $0 \le i < k$, let $x,y \in V$, and let $m = \max\{d(x,A_i), d(y,A_i), d(x,y)\}$. Then at least one of the following holds:

\begin{itemize}
\item $d_H(x,y) \le 6m$
\item $d_H(x,A_{i+1}) \le 10m$ and $d_H(y,A_{i+1}) \le 10m$.
\end{itemize}
\end{lemma}
\begin{proof}
Using the SLT rooted in $s_i$,  we have $d_H(x,p_i'(x))\le (1+\eps)d(x,A_i)$ and similarly $d_H(y,p_i'(y))\le (1+\eps)d(y,A_i)$, so that 
\begin{equation}\label{eq:qqq}
d(p_i'(x), p_i'(y)) \le d(p_i'(x),x) + d(x,y) + d(y,p_i'(y)) \le m + 2(1+\eps)m \le 3.2m~.
\end{equation}
The first case is that $p_i'(y) \in B_{\frac{1-\eps}{2}}(p_i'(x))$. Then by Lemma~\ref{lemma:distance-in-bunch}, $d_H(p_i'(x),p_i'(y))\le(1+\eps)d(p_i'(x), p_i'(y))$. 
\begin{align*}
d_H(x,y) &\le d_H(x,p_i'(x)) + d_H(p_i'(x), p_i'(y)) + d_H(p_i'(y), y) \\
&\stackrel{\eqref{eq:qqq}}{\le} (1+\eps)m + (1+\eps)(3.2m) + (1+\eps)m \\
&\le 6m.
\end{align*}
If $p_i'(y) \not\in B_{\frac{1-\eps}{2}}(p_i'(x))$ then 
\[
d(p_i'(x),p_{i+1}(p_i'(x)))\le \frac{2}{1-\eps} d(p_i'(x), p_i'(y)) \stackrel{\eqref{eq:qqq}}{\le} \frac{6.4m}{1-\eps}~.
\]
Now the second item holds, as
\begin{align*}
d_H(x,A_{i+1}) &\le (1+\eps)d(x,A_{i+1}) \\
&\le (1+\eps)[d(x,p_i'(x)) + d(p_i'(x), p_{i+1}(p_i'(x)))] \\
&\le (1+\eps)\left[(1+\eps)m + \frac{6.4m}{1-\eps}\right] \\
&\le 10m.
\end{align*}
\end{proof}
Now we can show the following:
\begin{lemma}\label{lemma:distance-spanner}
$H$ is an $(\alpha,\beta)$-spanner with $\alpha = 1+O(\eps)$ and $\beta = O(\frac{k}{\eps})^{k}$.
\end{lemma}
\begin{proof}
We set $\Delta = 7 + \frac{14k}{\eps}$. Let $x,y \in V$ and $W = W(x,y)$. If $d(x,y) \ge (3\Delta)^k W$, then the first item of Lemma~\ref{lemma:distance-long} must hold, as $A_{k+1}$ is empty by definition.
\[ d_H(x,y) \le \eta_{k}d(x,y) \le \left(1 + \eps + \frac{14k}{14k/\eps}\right) = (1+2\eps) d(x,y)\]
Otherwise $d(x,y) < (3\Delta)^{k}W$. Take the integer $i \in \{0,1,\ldots,k-1\}$ such that $(3\Delta)^i W \le d(x,y) < (3\Delta)^{i+1}W$. Apply Lemma~\ref{lemma:distance-long} on $x$, $y$, and $i$. If the first statement holds, we obtain $d_H(x,y) \le (1 + 2\eps)d(x,y)$ as before. Otherwise the second statement holds, in which $d_H(x,A_{i+1}) \le \frac{3\Delta}{\Delta - 7}d(x,y) \le 4d(x,y)$, and similarly $d_H(y,A_{i+1}) \le \frac{3\Delta}{\Delta - 7}d(x,y) \le 4d(x,y)$. These together imply $m \le 4d(x,y)$, with $m$ as in Lemma~\ref{lemma:distance-short}. Set $j = i+1$ and apply Lemma~\ref{lemma:distance-short} with $x$, $y$, and $j$.

If the first statement of Lemma~\ref{lemma:distance-short} holds with $x$, $y$, and $j$, then $d_H(x,y) \le 6m \le 24d(x,y)$. Otherwise, the second statement of Lemma~\ref{lemma:distance-short} holds, and it follows that \[
\max\{d(x,A_{i+1}), d(y,A_{i+1}), d(x,y)\} \le 10m.
\]
Repeatedly increase $j$ by 1 until the first statement of Lemma~\ref{lemma:distance-short} holds. The first statement must hold when $j = k$, in which we have $d_H(x,y) \le 24 \cdot 10^{k-i-1}d(x,y) \le 24 \cdot 10^{k-i-1}(3\Delta)^{i+1}W$. This is maximized at $i=k-1$, thus $d_H(x,y) \le 24 \cdot (3\Delta)^k W$. We obtain additive stretch at most $24 \cdot (21 + \frac{42k}{\eps})^k W = O\left(\frac{k}{\eps}\right)^k W$ as required.
\end{proof}

\subsection{Lightness of the spanner} \label{subsection:lightness}

To prove the desired lightness, we show that with high probability, the total lightness of the edges added in each phase does not exceed $\Oish(\frac{ n^{1/k}}{\eps})$. The following lemma bounds the lightness of $H_0$, created in the first phase.

\begin{lemma}\label{lemma:light-h0}
$\light(H_0) = O(\frac{\log n}{\eps})$.
\end{lemma}
\begin{proof}

Consider a fixed $i \in \{-1,0,\ldots,\log_2 n\}$ and recall that $N_{-1}=V$. Each vertex in $N_i \setminus N_{i+1}$ adds shortest paths of total length at most $2^{i+1} + 2^{i+2} + \ldots +  2^{i+\lceil\log(1/\eps)\rceil} = O(\frac{1}{\eps}) \cdot 2^i$. Recall that by Lemma~\ref{lem:MSTnet}, $|N|\Delta \le 2w(\MST(G))$ for any $\Delta$-net $N$. It follows that the total length of all shortest paths added corresponding to vertices in $N_i \setminus N_{i+1}$ is at most $O(\frac{1}{\eps}) w(\MST(G))$. We have $O(\log n)$ nets, so the lemma follows.
\end{proof}

In the second phase, we add a shortest path from each vertex to the representatives of its $\frac{1-\eps}{2}$-bunch. We start by bounding the total edge weight from all the paths to the representatives in the nets. Define for any $0\le j\le\log n$, $0\le i\le k-1$, and $x\in N_j$, the set $C_{ij}(x)=\{u\in A_i\setminus A_{i+1}~:~ \exists v\in B_{\frac{1-\eps}{2}}(u), x=r_j(v)\}$. This is the set of vertices $u\in A_i\setminus A_{i+1}$ that added a path to $x$ because this $x$ is a level-$j$ representative of some $v\in B_{\frac{1-\eps}{2}}(u)$.


\begin{lemma}\label{lemma:half-bunch-in-bunch}
Let $j \in \{0,\ldots,\log_2 n\}$, $x \in N_j$ and $0\le i\le k-1$. Then $\exists u\in C_{ij}(x)$ such that $C_{ij}(x) \subseteq B_1(u)$. 
\end{lemma}

\begin{figure}
    \centering
\tikzset{every picture/.style={line width=0.75pt}} 

\begin{tikzpicture}[x=0.75pt,y=0.75pt,yscale=-0.5,xscale=0.5]

\draw  [dash pattern={on 0.84pt off 2.51pt}] (212.1,159) .. controls (212.1,120.06) and (243.66,88.5) .. (282.6,88.5) .. controls (321.54,88.5) and (353.1,120.06) .. (353.1,159) .. controls (353.1,197.94) and (321.54,229.5) .. (282.6,229.5) .. controls (243.66,229.5) and (212.1,197.94) .. (212.1,159) -- cycle ;
\draw    (353.1,159) -- (282.1,159) ;
\draw [shift={(282.1,159)}, rotate = 180] [color={rgb, 255:red, 0; green, 0; blue, 0 }  ][fill={rgb, 255:red, 0; green, 0; blue, 0 }  ][line width=0.75]      (0, 0) circle [x radius= 2.34, y radius= 2.34]   ;
\draw [shift={(353.1,159)}, rotate = 180] [color={rgb, 255:red, 0; green, 0; blue, 0 }  ][fill={rgb, 255:red, 0; green, 0; blue, 0 }  ][line width=0.75]      (0, 0) circle [x radius= 2.34, y radius= 2.34]   ;
\draw   (125.28,159) .. controls (125.28,72.39) and (195.49,2.18) .. (282.1,2.18) .. controls (368.71,2.18) and (438.92,72.39) .. (438.92,159) .. controls (438.92,245.61) and (368.71,315.82) .. (282.1,315.82) .. controls (195.49,315.82) and (125.28,245.61) .. (125.28,159) -- cycle ;
\draw    (353.1,159) -- (395.2,122.6) ;
\draw [shift={(395.2,122.6)}, rotate = 319.15] [color={rgb, 255:red, 0; green, 0; blue, 0 }  ][fill={rgb, 255:red, 0; green, 0; blue, 0 }  ][line width=0.75]      (0, 0) circle [x radius= 2.34, y radius= 2.34]   ;
\draw [shift={(353.1,159)}, rotate = 319.15] [color={rgb, 255:red, 0; green, 0; blue, 0 }  ][fill={rgb, 255:red, 0; green, 0; blue, 0 }  ][line width=0.75]      (0, 0) circle [x radius= 2.34, y radius= 2.34]   ;
\draw  [dash pattern={on 0.84pt off 2.51pt}] (337.23,122.6) .. controls (337.23,90.58) and (363.18,64.63) .. (395.2,64.63) .. controls (427.22,64.63) and (453.18,90.58) .. (453.18,122.6) .. controls (453.18,154.62) and (427.22,180.58) .. (395.2,180.58) .. controls (363.18,180.58) and (337.23,154.62) .. (337.23,122.6) -- cycle ;
\draw    (353.1,159) -- (383.2,191.6) ;
\draw [shift={(383.2,191.6)}, rotate = 47.28] [color={rgb, 255:red, 0; green, 0; blue, 0 }  ][fill={rgb, 255:red, 0; green, 0; blue, 0 }  ][line width=0.75]      (0, 0) circle [x radius= 2.34, y radius= 2.34]   ;
\draw [shift={(353.1,159)}, rotate = 47.28] [color={rgb, 255:red, 0; green, 0; blue, 0 }  ][fill={rgb, 255:red, 0; green, 0; blue, 0 }  ][line width=0.75]      (0, 0) circle [x radius= 2.34, y radius= 2.34]   ;
\draw  [dash pattern={on 0.84pt off 2.51pt}] (344.21,191.6) .. controls (344.21,170.06) and (361.66,152.61) .. (383.2,152.61) .. controls (404.74,152.61) and (422.19,170.06) .. (422.19,191.6) .. controls (422.19,213.14) and (404.74,230.59) .. (383.2,230.59) .. controls (361.66,230.59) and (344.21,213.14) .. (344.21,191.6) -- cycle ;

\draw (353.1,162.4) node [anchor=north] [inner sep=0.75pt]    {$x$};
\draw (282.1,162.4) node [anchor=north] [inner sep=0.75pt]    {$u_{1}$};
\draw (200,270.3) node [anchor=north west][inner sep=0.75pt]    {$B_{1}( u_{1})$};
\draw (395.2,126) node [anchor=north] [inner sep=0.75pt]    {$u_{2}$};
\draw (383.2,195) node [anchor=north] [inner sep=0.75pt]    {$u_{3}$};

\end{tikzpicture}
\caption{Illustration of Lemma~\ref{lemma:half-bunch-in-bunch}. The dotted circles represent the $\frac{1-\eps}{2}$-bunches centered at $u_1$, $u_2$, and $u_3$ respectively. Lemma~\ref{lemma:half-bunch-in-bunch} states that all the $u_i$'s are contained in a 1-bunch of some vertex in $A_i$.}
\end{figure}
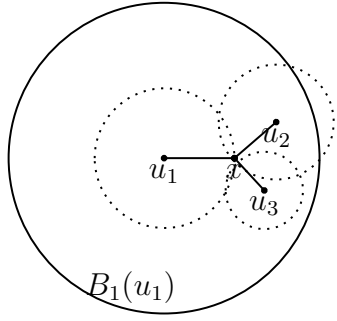

\begin{proof}
Let $u\in C_{ij}(x)$ be the vertex farthest from $x$.

 
As $x$ is a level $j$ representative of a vertex $v\in B_{\frac{1-\eps}{2}}(u)$, then by Lemma~\ref{lemma:representative} we have $d(v,x)\le d_H(v,x)\le (1+2\eps)2^j\le 2\cdot 2^j$. Recall that $j$ was chosen such that $2^j\le\frac{\eps}{4}d(u,v)$, hence
\begin{equation}\label{eq:xuv}
    d(u,x)\le d(u,v)+d(v,x)\le (1+\frac{\eps}{2})d(u,v)~.
\end{equation}

 It follows that for any $u'\in C_{ij}(x)$, 
 \begin{align*}
 d(u,u') &\le d(u,x)+d(x,u')\le 2d(u,x) \\
 &\stackrel{\eqref{eq:xuv}}{\le} 2(1+\frac{\eps}{2})d(u,v) \le 2(1+\frac{\eps}{2})(\frac{1-\eps}{2})d(u,A_{i+1}) < d(u,A_{i+1})~.
 \end{align*}
 Therefore by definition $u'\in B_1(u)$.
\end{proof}

Fix $j \in \{0,\ldots,\log_2 n\}$, $x \in N_j$, and $i \in \{0,\ldots,k-1\}$. We claim that each $u\in C_{ij}(x)$ adds a path to $x$ of weight $O(\frac{2^j}{\eps})$. To see this, let $v\in B_{\frac{1-\eps}{2}}(u)$ such that $x=r_j(v)$, then $2^j\ge\frac{\eps}{8}d(u,v)$, and by \eqref{eq:xuv} we have that $d(u,x)\le (1+\frac{\eps}{2})d(u,v)= O\left(\frac{2^j}{\eps}\right)$.

The following lemma was shown in \cite{TZ05}, we give a proof for completeness.
\begin{lemma}\label{lem:bunch-bound}
With high probability, for every $u\in V$, $|B_1(u)|\le O(n^{1/k}\ln n)$.
\end{lemma}
\begin{proof}
Let $0\le i\le k$ be such that $u\in A_i\setminus A_{i+1}$. Suppose first that $i\le k-1$. Then if we arrange all vertices in $A_i$ by non-decreasing order of distance from $u$ in $G$, each vertex in $A_i$ is promoted to $A_{i+1}$ independently with probability $n^{-1/k}$. The event that $|B_1(u)|\ge 2n^{1/k}\ln n$ is contained in the event that none of the first $2n^{1/k}\ln n$ vertices in this ordering is promoted. By independence, the probability of the latter event is
\[
(1-n^{-1/k})^{2n^{1/k}\ln n}\le e^{-2\ln n}=\frac{1}{n^2}~,
\]
so by taking a union bound over the $n$ vertices, with probability at least $1-\frac1n$ all bunches for $i<k$ are of size at most $O(n^{1/k}\ln n)$. Finally, note that for $u\in A_k$ we have $B_1(u)=A_k$. As each vertex is included in $A_k$ with probability $\frac1n$ independently, its expected size is 1, and by a simple Chernoff bound with high probability $|A_k|\le O(\ln n)$.
\end{proof}

Assume, as we may, that this high probability event did occur.
Thus for any representative $x\in N_j$, by Lemma~\ref{lem:bunch-bound}, there are only $O(n^{1/k}\log n)$ paths to $x$ from different $u\in C_{ij}(x)$ for each $0\le i\le k-1$. Furthermore, each of these paths has length $O(\frac{2^j}{\eps})$. 
We conclude that the total weight of paths connecting vertices in $V\setminus A_K$ to representatives in $N_j$ is $O\left(k\cdot |N_j|\cdot n^{1/k}\log n\cdot \frac{2^j}{\eps}\right)$. By Lemma~\ref{lem:MSTnet} we have that $|N_j|\cdot 2^j\le w(MST)$, so the lightness of these paths is $O\left(\frac{kn^{1/k}\log n}{\eps}\right)$. As we have only $O(\log n)$ different choices for $0\le j\le \log n$ we 
get lightness $O\left(\frac{kn^{1/k}}{\eps} \cdot \log^2n\right) = \Oish(\frac{n^{1/k}}{\eps})$.

It remains to bound the lightness from the paths between vertices in $A_k$. By our assumption $|A_k|\le O(\log n)$,  so there can be at most $O(\log^2n)$ vertex pairs in $A_k$. Each such pair may add one path to the spanner of total weight at most $w(MST(G))$, so the lightness of these paths is bounded by $O(\log^2n)$.

In the second phase we also add direct shortest paths between $u-v$ for any $v\in B_{\frac{1-\eps}{2}}(u)$ such that the unique integer $j$ satisfying 
$\frac{\eps}{8} \cdot d(u,v) \le 2^j < \frac{\eps}{4} \cdot d(u,v)$ is negative (that is, $v$ is its own representative). In this case, we have $d(u,v)=O(\frac{1}{\eps})$, and by Lemma~\ref{lem:bunch-bound}, the total number of such $u-v$ pairs is $O(n^{1+1/k}\log n)$ (once again, assuming the high probability event of the lemma). We conclude that the total lightness of the direct paths is $\Oish\left(\frac{n^{1/k}}{\eps}\right)$.


Finally, in the third phase we add the edges of $k$ $(1+\eps, O\left(\frac{1}{\eps}\right))$-SLTs to $H$.  Thus, the SLTs have $O(\frac{k}{\eps}) = O(\frac{\log n}{\eps})$ lightness.

Altogether, $H$ has lightness $O(\frac{\log n}{\eps})$ from the first and third phases of the construction (Lemma~\ref{lemma:light-h0}), and with high probability $(\frac{kn^{1/k}}{\eps} \cdot \log^2n) = \Oish(\frac{n^{1/k}}{\eps})$ lightness from the second phase. Therefore the lightness of $H$ is $\Oish(\frac{ n^{1/k}}{\eps})$.

\subsection{Size of the spanner}\label{subsection:size}
To show the size bound, we must first show that the expected number of edges added at each phase does not exceed $O(kn^{1 + 3/k})$. In the first phase, we add $O(n \log n)$ edges from the shortest paths to the net points; this follows as, for each $j \in \{0,\ldots,\log_2 n\}$, the set of edges added when connecting vertices to the net $N_j$ are a subset of the edges in a shortest path forest rooted at $N_j$. Since there are $O(\log n)$ nets, we obtain $O(n \log n)$ edges. Note that $n\log n\le kn^{1+3/k}$ for any value of $k\ge 1$.

We now bound the number of edges added in the second phase. The next lemma is an extension of a lemma in~\cite{elkin2019almost}:
\begin{lemma} \label{lemma:paths-intersect}
    Let $0 \le i \le k$. Let $u, x \in A_i$, and let $v$ (resp. $y$) be a representative of some vertex in $B_{\frac{1-\eps}{2}}(u)$ (resp. $B_{\frac{1-\eps}{2}}(x)$) for which a shortest $u$-$v$ path $P_{u,v}$ (resp. $x$-$y$ path $P_{x,y}$) is added to $H$. If $P_{u,v} \cap P_{x,y} \neq \emptyset$, then all four points are in $B_1(u)$, or all four points are in $B_1(x)$.
\end{lemma}
\begin{proof}
    Let $z\in V$ be a vertex in the intersection of $P_{u,v}$ and $P_{x,y}$. Assume first that $P_{u,v}$ is not shorter than $P_{x,y}$, and we show all points are in $B_1(u)$ (the other case is symmetric, in which all points are in $B_1(x)$). We have
    \begin{dmath*}    
    d(u,x) \le d(u,z) + d(z,x) \le 2d(u,v) \stackrel{\eqref{eq:xuv}}{\le} 2(1+\eps)\left(\frac{1-\eps}{2}\right)d(u,A_{i+1}) < d(u,A_{i+1}).
    \end{dmath*} 
Since $d(z,y)\le d(u,v)$ as well, the same bound holds for $d(u,y)$.
\end{proof}
For fixed $0 \le i \le k-1$, consider the graph $G_i$ containing all shortest paths $P_{u,v}$ where $u \in A_i$ and $v$ is the representative of some vertex in $B_{\frac{1-\eps}{2}}$ that $u$ connects to. The number of edges in $G_i$ is at most $O(n + C_i)$, where $C_i$ is the number of pairwise intersections between these paths.

\begin{lemma}
    $\mathbb{E}(|C_i|) = O(n^{(k+3-i)/k})$.
\end{lemma}
\begin{proof}
    By Lemma~\ref{lemma:paths-intersect}, for every intersecting pair of paths $P_{u,v}$ and $P_{x,y}$ in $G_i$, we have that all four points $u$, $v$, $x$, and $y$ belong to the same 1-bunch. Each $u \in A_i$ introduces at most $|B_1(u)|^3$ pairwise intersecting paths (since $u$ is fixed and $v, x, y \in B_1(u)$). Note that $|B_1(u)|$ is a random variable distributed geometrically with parameter $n^{-1/k}$, so
    \begin{align*}
        \mathbb{E}(|B_1(u)|^3) &= \sum_{j=1}^{\infty} j^3 n^{-1/k}(1 - n^{-1/k})^{j-1} = n^{-1/k} \sum_{j=1}^{\infty} j^3 (1 - n^{-1/k})^{j-1} \\
        &\le n^{-1/k} \sum_{j=1}^{\infty} (1 - n^{-1/k})^{j-1} j(j+1)(j+2) \le \frac{6}{n^{-3/k}} = O(n^{3/k}).
    \end{align*}
    We note that the expected size of $A_i$ is $n^{1-i/k}$, since every vertex joins $A_i$ with probability $n^{-i/k}$. The expected number of intersections at level $i$, and therefore the expected size of $G_i$, is
    \[ C_i = \mathbb{E}\left[ \sum_{u \in A_i} |B_1(u)|^3 \right] = O(n^{3/k}) \cdot n^{1-i/k} = O(n^{(k+3-i)/k}).\]
\end{proof}

For $i=k$, the expected size of $\binom{|A_k|}{2}$ is constant, so the expected number of edges added between vertices in $A_k$ is $O(n)$. Lastly, the $k$ SLTs add $O(kn)$ edges. The expected size of $H$ is $O(n \log n + n^{1+3/k} + kn) = O(kn^{1 + 3/k})$.

\subsection{Proof of Theorem~\ref{thm:spanner-near-additive}}
By the analysis of Section~\ref{subsection:stretch}, $H$ is a $(1+O(\eps), O(\frac{k}{\eps})^k W(\cdot,\cdot))$-spanner. Sections~\ref{subsection:lightness} and~\ref{subsection:size} showed that with high probability the lightness is $\Oish\left(\frac{n^{1/k}}{\eps}\right)$, and expected size of $H$ is $O(kn^{1+3/k})$.
Reducing the multiplicative stretch to $1+\eps$ can be done by scaling $\eps$ by an appropriate constant. So by Markov's inequality, with constant probability we get a spanner $H$ satisfying Theorem~\ref{thm:spanner-near-additive}.


\bibliographystyle{splncs04}
\bibliography{ref}

\appendix

\section{Proof of \Cref{lemma:distance-long}}\label{sec:lemma:distance-long:proof}

The proof is by induction on $i$. Let $W = W(x,y)$. For the base case $i=0$, if $y \in B_{\frac{1-\eps}{2}}(x)$, then $d_H(x,y) \le (1+\eps)d(x,y)$ by Lemma~\ref{lemma:distance-in-bunch}. This satisfies the first statement as $\eta_0=1+\eps$.

Otherwise, if $y \not\in B_{\frac{1-\eps}{2}}(x)$ and $x \in A_1$, then $d_H(x,A_1) = 0$ which satisfies the second statement. Lastly, if $x \in A_0 \setminus A_1$ and $y \not\in B_{\frac{1-\eps}{2}}(x)$, then $\frac{1-\eps}{2}d(x,p_1(x)) \le d(x,y)$. Using the SLT rooted at $s_1$, we have $d_H(x,A_1) \le (1+\eps)d(x,p_1(x)) \le \frac{2(1+\eps)}{1-\eps}d(x,y)$. As $\frac{2(1+\eps)}{(1-\eps)} < 3 < \zeta$, it follows that $d_H(x,A_1) \le \zeta d(x,y)$, so the second statement holds. Now assume Lemma~\ref{lemma:distance-long} holds for some $i \ge 0$ and we wish to prove for $i+1$. Let $x,y \in V$ such that $d(x,y) \ge (3\Delta)^{i+1}W$.

Divide the shortest path $P_{x,y}$ into $J$ segments $L_j = [u_j, u_{j+1}]_{j \in [J]}$ of length at least $(3\Delta)^iW$ and at most $d(x,y)/\Delta$ using the following method: let $u_1 = x$, and for $j \ge 1$, let $u_{j+1}$ be the first vertex on $P_{x,y}$ such that the length of $u_ju_{j+1}$ is at least $(3\Delta)^i W$. If there is no such vertex $u_{j+1}$, set $y = u_{j+1}$ and $|J|=j$. Combine the last two segments to create one segment. Each segment has length at most $(3\Delta)^i W + W = ((3\Delta)^i + 1)W$, except for the last segment, which has length at most $(3\Delta)^iW + W + (3\Delta)^i W \le 3^{i+1}\Delta^i W \le d(x,y)/\Delta$. If the first statement of Lemma~\ref{lemma:distance-long} holds for all segments in $L_j$ with parameter $i$, then the first statement holds for the pair $(x,y)$ with parameter $i+1$:

\begin{align}\label{eq:pooo}
d_H(x,y) &\le \sum_{j \in [J]} d_H(u_j, u_{j+1}) \le \eta_i \sum_{j \in [J]} d(u_j, u_{j+1}) = \eta_i d(x,y) \le \eta_{i+1} d(x,y).
\end{align}

Otherwise, let $L_l = [u_l,u_{l+1}]$ and $L_{r-1}=[u_{r-1},u_r]$ be the leftmost and rightmost segments for which the first statement is false and the second statement is true; note these may be the same segment. That is,
\begin{align}\label{eq:secco}
d_H(u_l, A_{i+1}) &\le \zeta d(u_l, u_{l+1}) \le \frac{\zeta}{\Delta}d(x,y) \\
d_H(u_r, A_{i+1}) &\le \zeta d(u_{r-1}, u_r) \le \frac{\zeta}{\Delta}d(x,y)\nonumber
\end{align}
Recall that $p_i'(u)$ is the ``approximate'' pivot $u$ connects to, namely, the vertex in $A_i$ that the SLT rooted at $s_i$ connects $u$ to, such that $d_H(u,p_i'(u)) \le (1+\eps)d(u,p_i(u))$. Consider first the case that $p_{i+1}'(u_r) \in B_{\frac{1-\eps}{2}}(p_{i+1}'(u_l))$, 
by Lemma~\ref{lemma:distance-in-bunch}:
\begin{align}
d_H(p_{i+1}'(u_l), p_{i+1}'(u_r)) &\le (1+\eps)d(p_{i+1}'(u_l), p_{i+1}'(u_r)) \tag*{}\\
&\le (1+\eps)[d(p_{i+1}'(u_l), u_l) + d(u_l,u_r) + d(u_r, p_{i+1}'(u_r))] \label{eq:distance-between-pivots}
\end{align}
We can bound the distance between $u_l$ and $u_r$ in $H$ by
\begin{dmath*}
d_H(u_l,u_r) \le  d_H(u_l, p_{i+1}'(u_l)) + d_H(u_r,p_{i+1}'(u_r)) + d_H(p_{i+1}'(u_l), p_{i+1}'(u_r)) \\
\stackrel{\eqref{eq:distance-between-pivots}}{\le}  d_H(u_l, p_{i+1}'(u_l)) + d_H(u_r,p_{i+1}'(u_r)) + (1+\eps)[d(p_{i+1}'(u_l), u_l) + d(u_l,u_r) + d(u_r, p_{i+1}'(u_r))] \\
\le  (1+\eps)d(u_l,u_r) + (2+\eps)[d_H(u_l,p_{i+1}'(u_l)) + d_H(u_r,p_{i+1}'(u_r))].
\end{dmath*}
By the SLT and \eqref{eq:secco} we have $d_H(u_l, p_{i+1}'(u_l)) \le (1+\eps)d_H(u_l, A_{i+1}) \le \frac{(1+\eps)\zeta}{\Delta} d(x,y)$. Similarly, $d_H(u_r,p_{i+1}'(u_r)) \le \frac{(1+\eps)\zeta}{\Delta} d(x,y)$. This implies that
\begin{align}\label{eq:ttt}
d_H(u_l,u_r) &\le (1+\eps)d(u_l,u_r) + \frac{(2+\eps)(2+2\eps)\zeta}{\Delta}d(x,y) \nonumber\\
&\le \eta_i d(u_l,u_r) + \frac{4.62\zeta}{\Delta}d(x,y) \nonumber \\
&\le \eta_i d(u_l,u_r) + \frac{14}{\Delta-7}d(x,y).
\end{align}

By the same calculation as in \eqref{eq:pooo}, we have that $d_H(x,u_l)\le\eta_id(x,u_l)$ and $d_H(u_r,y)\le\eta_id(u_r,y)$.
We can finally bound $d_H(x,y)$ as follows.

\begin{align*}
d_H(x,y) &\le d_H(x,u_l) + d_H(u_l,u_r) + d_H(u_r,y) \\
&\stackrel{\eqref{eq:ttt}}{\le} \eta_i d(x,u_l) + \left[\eta_i d(u_l,u_r) + \frac{14}{\Delta-7}d(x,y)\right] + \eta_i d(u_r,y) \\
&\le \left(\eta_i + \frac{14}{\Delta-7}\right)d(x,y) \\
&= \eta_{i+1}d(x,y).
\end{align*}
The last step follows as $\eta_i$ was chosen to satisfy $\eta_i + \frac{14}{\Delta-7} = \eta_{i+1}$.

The second case is that $p_{i+1}'(u_r) \not\in B_{\frac{1-\eps}{2}}(p_{i+1}'(u_l))$. Then by definition of bunch,
\begin{equation}\label{eq:not-in-bunch}
d(p_{i+1}'(u_l), p_{i+2}(p_{i+1}'(u_l))) \le \frac{2}{1-\eps}d(p_{i+1}'(u_l), p_{i+1}'(u_r))
\end{equation}
Note that because there is an SLT path in $H$ from $x$ to some vertex in $A_{i+2}$, we have $d_H(x,A_{i+2}) \le (1+\eps)d(x,A_{i+2})$. We upper bound $d_H(x,A_{i+2})$ as follows:
\begin{dmath*}
d_H(x,A_{i+2}) \le  (1+\eps)d(x,A_{i+2}) \\
\le  (1+\eps)\left[d(x,u_l) + d(u_l, p_{i+1}'(u_l)) + d(p_{i+1}'(u_l), p_{i+2}(p_{i+1}'(u_l)))\right] \\
\stackrel{\eqref{eq:not-in-bunch}}{\le}  (1+\eps)[d(x,u_l) + d(u_l, p_{i+1}'(u_l)) + \frac{2}{1-\eps}d(p_{i+1}'(u_l), p_{i+1}'(u_r))] \\
\le  (1+\eps)\left[d(x,u_l) + d(u_l, p_{i+1}'(u_l)) +  \frac{2}{1-\eps}\left(d(p_{i+1}'(u_l),u_l) + d(u_l,u_r) + d(u_r,p_{i+1}'(u_r)) \right)\right] \\
=  (1+\eps)\left[d(x,u_l) + \frac{2}{1-\eps}d(u_l,u_r)\right] + (1+\eps)\left[d(u_l,p_{i+1}'(u_l)) + \frac{2}{1-\eps}\left(d(p_{i+1}'(u_l), u_l) + d(u_r, p_{i+1}'(u_r))\right)\right] \\
\le  \frac{2(1+\eps)}{1-\eps}d(x,y) + (1+\eps)\left(1 + \frac{4}{1-\eps}\right)\max\{d(u_l,p_{i+1}'(u_l)),d(u_r,p_{i+1}'(u_r))\} \\
\le  \frac{2(1+\eps)}{1-\eps}d(x,y) + (1+\eps)^2\left(1 + \frac{4}{1-\eps}\right)\max\{d(u_l,A_{i+1}),d(u_r,A_{i+1})\} \\
\stackrel{\eqref{eq:secco}}{\le}  \left(\frac{2(1+\eps)}{1-\eps} + (1+\eps)^2\left(1 + \frac{4}{1-\eps}\right)\frac{\zeta}{\Delta}\right) d(x,y) \\
\stackrel{(\eps < \frac{1}{10})}{\le}  \left(3 + \frac{7\zeta}{\Delta}\right)d(x,y) \\
=  \zeta d(x,y).
\end{dmath*}
The last equality follows as $\zeta$ was chosen to satisfy the equation $3 + \frac{7\zeta}{\Delta} = \zeta$.

\section{A sparse and lightweight \texorpdfstring{$(1+\eps, 2(1+\eps)W_{max})$}~-spanner}
In this section we prove the following theorem.
\begin{theorem} \label{thm:spanner-1-2}
Let $G$ be a weighted graph such that $W_{max} \ge \sqrt{w(\MST(G))}$, and let $\eps > 0$. Then $G$ has a $(1+\eps, 2(1+\eps)W_{max})$-spanner of size $O(n^{3/2})$ and lightness $O(\frac{n^{1/2}}{\eps})$.
\end{theorem}


\begin{proof}
We use the following simple construction using Definitions~~\ref{def:slt} and~\ref{def:delta-net}. Assume w.l.o.g. $w(\MST(G)) = n$, so that $W_{max} \ge \sqrt{n}$. Let $N$ be a $\sqrt{n}$-net. For each $v \in N$, add a $(1+\eps, O(\frac{1}{\eps}))$-SLT rooted at $v$; let $H$ be the resulting graph. Then for each vertex pair $(x,y) \in V \times V$, there is a vertex $v \in N$ such that $d(x,v) \le \sqrt{n} \le W_{max}$. The distance $d_H(x,y)$ between $x$ and $y$ is at most
\begin{align*}
d_H(x,y) &\le d_H(x,v) + d_H(v,y) \\
&\le (1+\eps)(d(x,v) + d(v,y)) \\
&\le (1+\eps)(\sqrt{n} + (\sqrt{n} + d(x,y))) \\
&= (1+\eps)d(x,y) + 2(1+\eps)W_{max}
\end{align*}
By Lemma~\ref{lem:MSTnet} we have  $|N|\sqrt{n} \le 2w(\MST(G))$, it follows that $|N| = O(\sqrt{n})$, and for each $v \in V$, we add an SLT of weight $O(\frac{1}{\eps}) \cdot w(\MST(G))$. We conclude that the lightness of $H$ is $O(\sqrt{n} \cdot \frac{1}{\eps})$. Moreover, the size of the spanner is $O(n^{3/2})$ since $H$ consists of $O(\sqrt{n})$ spanning trees.
\end{proof}

\end{document}